\documentclass[a4paper,10pt]{article}

\usepackage{amsthm}{\normalsize }
\usepackage{amsmath}
\usepackage{amssymb}
\usepackage{latexsym}

\usepackage[textwidth=13cm,textheight=22cm]{geometry}

\newtheorem{theorem}{Theorem}[section]
\newtheorem{proposition}[theorem]{Proposition}
\newtheorem{lemma}[theorem]{Lemma}
\newtheorem{corollary}[theorem]{Corollary}
\newtheorem{remark}[theorem]{Remark}

\newtheorem{example}[theorem]{Example}
\newtheorem{assumption}[theorem]{Assumption}

\newcommand{\FF}{\mathcal{F}}
\newcommand{\f}{\mathbb{F}}
\newcommand{\R}{\mathbb{R}}
\newcommand{\PP}{\mathbb{P}}

\newcommand{\e}{\varepsilon}

\begin{document}

\title{Sticky processes, local and true martingales\thanks{The first author was supported by the
``Lend\"ulet'' Grant LP2015-6 of the Hungarian Academy of Sciences. Discussions with
Martin Keller-Ressel led to formulating the main results of the present paper, we sincerely
thank him. We also thank Eberhard Mayerhofer for spotting an error and two anonymous referees
for helpful reports which revealed, in particular, another problem in a previous version
of this paper.}}

\author{Mikl\'os R\'asonyi\thanks{MTA Alfr\'ed R\'enyi Institute of Mathematics, Re\'altanoda utca 13-15,
1053 Budapest, Hungary. E-mail: rasonyi@renyi.mta.hu.}
\and Hasanjan Sayit\thanks{Department of Mathematical Sciences, Durham University, South Road, Durham DH1 3LE, United Kingdom. E-mail: hasanjan.sayit@durham.ac.uk.}}

\date{\today}

\maketitle

\begin{abstract} 
We prove that for a so-called sticky process $S$ there exists an equivalent probability $Q$
and a $Q$-martingale $\tilde{S}$ that is arbitrarily close to $S$ in
$L^p(Q)$ norm. For  continuous $S$, $\tilde{S}$ can be chosen arbitrarily close to $S$ in  supremum
norm. In the case where $S$ is a local martingale 
we may choose $Q$ arbitrarily close to
the original probability in the total variation norm.
We provide examples to illustrate the power of our results and present an application in mathematical
finance.
\end{abstract}

\section{Introduction}

%In this note we show that stochastic processes that satisfy the natural condition of stickiness can be approximated arbitrarily closely by semi-martingales that admit equivalent martingale measures  under the $L^p$ norm. A process is sticky if it never certain to exit a small ball in a finite time horizon no matter how small the ball is.  So far many interesting classes of stochastic processes are shown to have the stickiness property. In proposition 3.1 of Guasoni \cite{sticky}, regular strong Markov processes are shown to have the stickiness property. This includes almost all L\'evy processes (the only L\'evy processes that are not sticky are the deterministic ones as will be discusses in an example later).

By their very definition, local martingales are
``almost'' martingales. Moreover, in discrete
time every local martingale is a martingale under
an equivalent change of measure and
the new measure can be chosen to be arbitrarily
close to the original one in the total variation norm, even on an infinite horizon,
see e.g. Theorem 2.2.2 in Kabanov and Safarian \cite{kabanov}.

In continuous time such a strong result does not hold. For example  the inverse of the three dimensional  Bessel
process is a local martingale and it is not a martingale under any
equivalent change of probability measure.
We may ask, however, whether there is a process
``near'' the given local martingale which becomes a martingale
under an equivalent probability.

It turns out that such a result holds provided that the given local martingale satisfies
the natural condition of \emph{stickiness}: for sticky local martingales a martingale (modulo a 
change of measure to some $Q\sim P$)  that stays in any small neighborhood of it under the $L^p(Q)$ norm  
can be found, and $Q$ can even be 
chosen to be as close as one wants to $P$ in total variation norm, see Corollary 
\ref{lp} below for this result.

A process is sticky if, starting from any stopping time on,  it is never certain to exit a small ball in 
a given time horizon no matter how small the ball is. This condition was first used in the paper Guasoni 
\cite{sticky} in the context of finance and according to the Proposition 3.1 of Guasoni \cite{sticky} all 
regular strong Markov processes are sticky. This includes, for example,  most L\'evy processes, see 
Section \ref{exa} for further details.  Other than this, stochastic processes with the
conditional full support (henceforth,  CFS) property are also sticky. The CFS property (see Remark \ref{utal}
below for its definition)
was introduced in the paper Guasoni et al. \cite{grs} and a large class of stochastic processes, 
including fractional Brownian motion (fBm), enjoys this  property, see \cite{cherny, GSZ, HPR, mikko}  for example.

In Guasoni et al. \cite{grs},  it was  shown that processes with CFS can be approximated arbitrarily closely under the supremum norm  by semi-martingales that admit equivalent martingale measures. In the subsequent  paper Bender et al. \cite{bps},  the same result  was obtained for continuous path processes that are merely sticky. In these papers, such approximation  was possible because the stochastic processes 
were assumed to be continuous. 

For the case of jump processes, approximation under the supremum norm, however,  seems difficult if 
not impossible. In this note we show, along with our result on local martingales, that  c\`adl\`ag sticky processes can be approximated by 
martingales (modulo a change of measure to some $Q\sim P$) arbitrarily closely under the $L^p(Q)$ norm. 
%Not only this, we actually show a stronger result: we can choose the martingale to be equal to the 
%original process at any given discrete set of time points, see Remark \ref{subinterval} below.  

 %c\`adl\`ag processes can be approximated by martingales (modulo an equivalent change of measure) under the $L^p$ norm arbitrarily closely.

%In the strongest possible sense the processes should be uniformly close along almost every 
%trajectory. Indeed, for certain continuous local martingales this can be achieved, see Theorem \ref{tetto} below.
%For local martingales with jumps, however, we should be contented with somewhat weaker
%forms of closeness e.g. in $L^p$ norms, see Corollary \ref{lp} below.

The paper is organized as follows. In Section \ref{stick} we recall the stickiness condition.
In Section \ref{exa}, we provide examples of sticky processes. In Section \ref{mr}
we prove that sticky processes can be approximated ``arbitrarily closely''
by  martingales in the sense explained above, see Theorem \ref{main1}
and Corollary \ref{added}. In Section
\ref{loca} we show that,  in the case of local martingales, one can choose the new
probability measure arbitrarily close to the original one in the total variation norm, see Theorem \ref{tetto}
and Corollary \ref{lp}.  In Section
\ref{appli},  we explain the relevance of our results to mathematical finance.
Finally, some technical details are relegated to Section \ref{appendix}.

\section{Sticky processes}\label{stick}

Let $(\Omega,\mathcal{F}, \PP)$ be a probability space. Let $S=(S_t)_{t\in [0,T]}$ be a c\`adl\`ag $\mathbb{R}^d$-valued process adapted to a filtration
$\f=(\mathcal{F}_t)_{t\in [0,T]}$ satisfying the usual assumptions 
(i.e., $\f$ is right continuous and $\FF_0$ contains all of the $\PP$ null sets of $\FF$). 
In this paper, for generality's sake, we do not assume that $\FF_0$ is a trivial $\sigma-$algebra. 
It can contain sets other than just the null and full measure sets.

%on this probability space. 
We say that the process $S$ is \emph{sticky} with respect to the filtration $\f$ if, for any stopping time $\tau$ of $\f$ and any $\FF_{\tau}-$measurable strictly positive random variable $\kappa$, the following condition is satisfied
\begin{equation}\label{stickiness0}
\PP(\sup_{u\in [\tau,T]} |S_u-S_{\tau}|<\kappa|\mathcal{F}_{\tau})>0\mbox{ a.s.} 
\end{equation}
Here $|\cdot|$ is the Euclidean norm of $\R^d$.  This definition is clearly equivalent to Definition 2.2 of 
Guasoni \cite{sticky} where $\kappa$ is assumed to be any deterministic number. In Lemma 3.1 of 
Bender et al. \cite{bps} it was shown that, for processes with continuous paths, stickiness  is equivalent to 
\begin{equation}\label{stickiness1}
\PP(\sup_{u\in [t,T]} |S_u-S_{t}|<\kappa|\mathcal{F}_{t})>0\mbox{ a.s.}, 
\end{equation}
for any deterministic time point  $0\le t\le T$ and any strictly positive and $\FF_{t}$-measurable random variable $\kappa$.  Lemma 3.1 of Bender et al. \cite{bps} is 
also true for c\`adl\`ag processes, this is the content of Lemma \ref{stickiness2} below. 

\begin{lemma}\label{stickiness2} A  c\`adl\`ag process $S$ is sticky iff it satisfies 
(\ref{stickiness1}) for any deterministic  $t\in [0,T]$.
\end{lemma}
\begin{proof} One direction is trivial. To show the other one, let $0\le \tau \le T$ 
be any stopping time of $\mathbb{F}$. Let $\kappa$ be any strictly positive  $\FF_{\tau}-$measurable random variable. Take any $A\in \FF_{\tau}$ with $P(A)>0$.
We would like to show that
\[
P(A\cap \{\sup_{t\in [\tau, T]}|S_{t}-S_{\tau}|<\kappa \})>0.
\]
Without loss of generality assume that 
$\tau<T$ on $A$. There exists a (deterministic)  
rational number $r>0$ such that
\[
A_r:=A\cap \left\{\sup_{t\in [\tau, r]}|S_t-S_{\tau}|<\frac{\kappa}{2}\right\}\cap\left\{\tau\leq r\right\}
\]
has positive probability. This can be seen from the following, obvious relation:
\[
A=\bigcup_{r\in [0, T]\cap \mathbb{Q}} 
\left(\left\{\tau\leq r\right\}\cap\left\{\sup_{t\in [\tau, r]}|S_t-S_{\tau}|<\frac{\kappa}{2}\right\}\cap A \right),
\]
which holds since $S$ is right-continuous. Observe that $A_r\in \FF_r$ so (\ref{stickiness1}) implies that
\[
P\left(A_r\cap \left\{\sup_{t\in [r, T]}|S_t-S_r|<\frac{\kappa}{2}\right\}\right)>0.
\]
Now the claim follows from
\[
A_r\cap \left\{\sup_{t\in [r, T]}|S_t-S_r|<\frac{\kappa}{2}\right\} \subset A\cap \left\{\sup_{t\in [\tau, T]}|S_{t}-S_{\tau}|<\kappa 
\right\}.
\]
\end{proof}

\begin{remark}{\rm For Markov processes $S$ stickiness reduces to checking 
\[
P(\sup_{u \in [t, T]}|S_{u}-S_t|<\kappa|S_t)>0
\] 
for almost all $\omega$ and all $\kappa>0$, $0\leq t<T$. For processes with independent increments, 
it boils down to $P(\sup_{u\in [t,T]}|S_u-S_t|<\kappa)$
being positive for all $\kappa>0$, $0\leq t<T$. It follows thus from Simon \cite{simon} 
that most L\'evy processes have the stickiness property, see Example \ref{nonny} below. See also Aurzada and Dereich \cite{AD} for
more results on the related theory of ``small deviations''.
}
\end{remark}

\begin{remark}\label{utal}{\rm Processes with the CFS property in any open domain are sticky. We recall the CFS property here.
Let $O$ be a non-empty open subset of $\mathbb{R}^d$ and let $C[a,b](O)$ denote the metric space of 
$O$-valued continuous functions on the interval $[a,b]$ equipped with the metric coming from the supremum norm.
For $x\in O$, set $C_x[a,b](O):=\{f\in C[a,b](O):\,f(a)=x\}$. We say that $S$ has \emph{conditional full support}
in $O$
(CFS-$O$) if $S$ has continuous trajectories in $O$ and
for all $0\leq t<T$, 
$$
\mathrm{supp}\, P(S\vert_{[t,T]}\in \cdot\vert\mathcal{F}_t)=C_{S_t}[t,T](O).
$$
Here $P(S\vert_{[t,T]}\in \cdot\vert\mathcal{F}_t)$ denotes the $\mathcal{F}_t$-conditional
distribution of the $C[t,T](O)$-valued random variable $S\vert_{[t,T]}$. When $O=\mathbb{R}^d$
we simply write CFS instead of CFS-$O$.}
\end{remark}

%If $S$ has CFS-$O$ then it
%is easily seen to be sticky. For instance, fractional Brownian motion always has CFS,
%see Guasoni et al. \cite{grs}. In Cherny \cite{cherny} it 
%was shown that Gaussian moving average processes have CFS. In the relevant papers Gasbarra et al. \cite{GSZ} 
%and Pakkanen \cite{mikko}, certain other classes of processes are shown to have the CFS property.

%For processes with independent increments, stickiness boils down to $P(\sup_{t\in [u,T]}|S_t-S_u|<\kappa)$
%being positive for all $\kappa>0$, $0\leq u<T$,
%by Lemma \ref{stickiness2}. It follows thus from Simon \cite{simon} 
%that most L\'evy processes have the stickiness property, see also Aurzada and Dereich \cite{AD} for
%more on the related theory of ``small deviations''.
\begin{remark} \label{trans}{\rm
Stickiness is invariant under composition with continuous functions and, in the case of $S$ with
continuous trajectories, under bounded time changes, as shown 
in Sayit and Viens \cite{sv}. This helps to generate a large class of sticky processes. 
For example, the process $|B_t|^{\frac{1}{3}}$, where $B_t$ is a one dimensional Brownian motion, is 
not a semimaringale according to Theorem 72 on page 221 of Protter \cite{pro} though it is sticky. 
See Section \ref{exa} for further examples.}
\end{remark}

In the recent paper Bender et al. \cite{bps}, it was shown that if a continuous path process is 
sticky then for any $\e>0$  there exists a semi-martingale $\tilde{S}$ that admits an equivalent martingale 
measure such that 
\begin{equation}\label{beki}
\sup_{t\in [0, T]}|S_t-\tilde{S}_t|<\e
\end{equation}  
holds almost surely. To prove their main result, they constructed a discrete time stochastic sequence
that is sufficiently close to the stochastic sequence obtained by stopping the process at 
each $\varepsilon$-increments and that, in the meantime,  satisfies the conditions of Theorem 2.1 in Kabanov and Stricker \cite{ks}. They were able to show that the sets 
$C_n^i, n\in N, i=1, 2, \cdots, 2d+1$ defined in their paper have positive conditional probabilities, see Lemma 3.3 of that paper.  A closer  look reveals that the continuous path property of the stochastic processes plays a key role in the proof of this Lemma 3.3. In the presence of jumps, we can not obtain the same property for the sets $C_n^i$ as in their Lemma 3.3. However, we are able to prove a similar result for sticky jump processes under an additional assumption which will be stated below. Also, in the presence of jumps, we can only control the \emph{moments}
of the supremum in \eqref{beki}. The following is the assumption that we will need in the proof of our main results.

%We assume that the probability space is large enough
%in the following sense.

\begin{assumption}\label{large}
The probability space supports a $d$-dimensional Brownian motion $B_t$, $t\in [0,T]$ with its augmented
natural filtration $\mathbb{G}=(\mathcal{G}_t)_{t\in [0, T]}$ such that $\mathcal{G}_T$ is independent
of $\mathcal{F}_T$.  
\end{assumption}

\begin{remark} {\rm Such an assumption often appears in stochastic analysis, e.g. 
recall the theorem asserting that a continuous martingale is a time-changed Brownian motion. 
In the present setting, we use this extra Brownian motion to construct a new sticky process 
which is as close as we want to 
the original process and has a sufficiently rich collection of paths. We then use this new sticky process to construct 
the $Q\sim P$ and $\tilde{S}$ we want, see Theorem \ref{main1} below.} 
\end{remark}

\section{Examples}\label{exa}

In this section, we give some examples of sticky processes. As stickiness is invariant under 
various transformations with continuous functions, identifying the stickiness property for 
stochastic processes,  even when they admit martingale measures,  is useful. Most L\'evy processes are known to admit equivalent martingale measures,
see Proposition 9.9 on page 315 of \cite{rama}, for example. However, their transformations under continuous functions may  lose even the semi-martingale property as discussed in 
Remark \ref{trans}.

\begin{example}\label{w} {\rm
Let $W$ denote a $d$-dimensional Brownian motion. Let $b:\mathbb{R}^d\to\mathbb{R}^d$ be locally bounded and
$v:\mathbb{R}^d\to\mathbb{R}^{d\times d}$ be continuous with $v(x)$ non-singular for
all $x\in\mathbb{R}^d$. If the stochastic differential equation
$$
dX_t=b(X_t)dt+v(X_t)dW_t,\ X_0=x,
$$
has a weak solution, unique in law, for all $x\in\mathbb{R}^d$,  then any solution satisfies CFS, a fortiori,
stickiness, as shown
in Guasoni and R\'asonyi \cite{fragile}. 

CFS also holds for many non-semimartingales: fractional Brownian motion and other
Gaussian processes, see \cite{grs,cherny, GSZ}.}
\end{example}

\begin{example}\label{w1} {\rm Let's look at the case of a skew Brownian motion $X_t$ which is defined to be the solution of the following equation
\[
X_t=W_t+\beta L_t^0,
\]
where $W_t$ is a one-dimensional Brownian motion,  $L_t^0$ is local time of the unknown process 
$X_t$ at time $0$, and $\beta$ is a constant with $|\beta|<1$, see Harrison and Shepp \cite{shepp} 
for further  details. Since the local time $L_t^0$ generates a measure singular to the Lebesque measure,  
$X_t$  does not admit any local martingale measure. Let $\alpha=(\beta+1)/2$ and define the strictly 
monotone continuous function $s_{\alpha}$  as  $s_{\alpha}=(1-\alpha)x$ for $x\geq 0$, $\alpha x$ for 
$x<0$ . Let $Y_t=s_{\alpha}(X_t)$. It was shown in Harrison and Shepp \cite{shepp} that 
$Y_t$ satisfies $dY_t=f(Y_t)dW_t$, where $f(x)=1-\alpha$ for $x>0$, $\frac{1}{2}$ for $x=0$, and $\alpha$ 
for $x<0$. Since $f$ is non-singular and bounded,  from the results of Stroock and Varadhan \cite{stva} we 
can conclude that $Y_t$ has full support on the space of continuous functions for any initial value.
Consequently, as $Y_t$ is Markovian, it has CFS and hence it is sticky. It is clear that $Y_t$ 
is a martingale as $f$ is bounded. The process $X_t$ inherits the stickiness property from 
$Y_t$ as it can be written as a composition of $Y_t$ with a strictly monotone continuous function
(the inverse function of $s_{\alpha}$).}
\end{example}

\begin{example}\label{nonny}{\rm 
Let us turn to processes with jumps now. For simplicity we assume $d=1$. 
Let $L$ be a L\'evy process. Then it has the following decomposition
\begin{eqnarray}\label{levy}
L_t=ct+\sigma B_t+\int_{|\theta |<1}\theta \tilde{N}(t, d\theta)+\int_{|\theta|\geq 1}\theta N(t, d\theta),
\end{eqnarray}
for some constants $c, \sigma \in \mathbb{R}$. Here $\nu$ is the L\'evy measure of $L$,
 $N$ the Poisson random measure of $L$,  and  
$\tilde{N}(dt, d\theta)=N(dt, d\theta)-\nu (d\theta)dt$ is its compensated version. $B$ is an independent Brownian motion from $N$.
 %with 
%triplet $(\nu,\sigma^2,c)$, where $\nu$ is the L\'evy-measure, $\sigma^2$ is the
%variance of the Brownian component and $c$ is the coefficient of the linear component. 
It is shown in Simon \cite{simon} that $L$ satisfies the stickiness property provided that
$\sigma^2\neq 0$ or $\int_{-1}^1 |x|\nu(dx)=\infty$. If $\sigma^2=0$ and 
$\int_{-1}^1 |x|\nu(dx)<\infty$ then $L$ satisfies stickiness if $h:=c-\int_{-1}^1 |x|\nu(dx)=0$
or $h>0$ (resp. $h<0$) and, for all $\epsilon>0$, $\nu((-\epsilon,0))>0$
(resp. $\nu((0,\epsilon))>0$). }
\end{example}

\begin{example}\label{w2}{\rm Let $X$ satisfy CFS and let $L$ be a sticky 
L\'evy process such that they are independent. 
Then $S_t:=f(X_t,L_t)$ is  also sticky for any continuous function $f:\mathbb{R}^{d+1}\to\mathbb{R}$,
by Proposition 1 of Sayit and Viens
\cite{sv}. For example, one can replace the Brownian motion 
$B_t$ in (\ref{levy}) by a fractional Brownian motion $B_t^H$ that is independent from $N_t$,  
and obtain a sticky process which is not a semi-martingale.}
\end{example}

\begin{remark}{\rm
We expect that solutions of L\'evy process-driven stochastic differential equations
are also sticky under mild conditions. It is outside the scope of the present paper
to pursue related investigations.}
\end{remark}

\section{Main result}\label{mr}

As explained in the above section, a large class of stochastic processes enjoy the stickiness property. 
Our main goal in this section is to show that martingales (under an equivalent change of measure) live ``near'' to them 
e.g. in the $L^p$ norm. In the following Theorem we state this result and present its proof after some preparations.

\begin{theorem}\label{main1} Let $g:\mathbb{R}_+\to\mathbb{R}_+$ be a convex function with 
$g(0)=0$ and let $\chi>0$ be any fixed number. Let $S$ be a c\`adl\`ag
process which is sticky with respect to $\f$. Let Assumption
\ref{large} be in force. Let $\mathcal{H}_t=\FF_t\vee \mathcal{G}_t$ for each $t\in [0, T]$. Then the process 
$S$ is sticky with respect to $\mathbb{H}=(\mathcal{H}_t)_{t\in [0, T]}$ and  there exists $Q\sim P$ and a 
$d$-dimensional $Q$-martingale $\tilde{S}$ (with respect to $\mathbb{H}$) such that $\tilde{S}_0=S_0$ and
\begin{equation}\label{labda}
E_Qg(\sup_{t\in [0,T]}|S_t-\tilde{S}_t|)<\chi.
\end{equation}
If $S$ has continuous trajectories
then even 
\begin{equation}\label{reffi}
\sup_{t\in [0,T]}|S_t-\tilde{S}_t|<\chi
\end{equation}
holds a.s. 
%If $S$ is (strictly) positive then
%so is $\tilde{S}$.
\end{theorem}

\begin{example}\label{alma}
 {\rm In general, it is not possible to replace $Q$ by the physical measure $P$ in 
\eqref{labda} above. This is shown by a simple example:
let $T:=1$, $S_t:=0$, $t<1$,  and let $S_1$ be uniform on $[0,1]$. We take $\mathbb{F}$ to be 
the natural filtration of $S$.
Set $g(x):=|x|$ and choose $\chi:=1/4$.

The process $S$ is trivially sticky. Arguing by contradiction, suppose that there is $\tilde{S}_1$ such that 
$\chi>E\sup_{t\in [0,1]} |E[\tilde{S}_1\vert
\mathcal{H}_t]-S_t|$. Then also $E|E[\tilde{S}_1\vert
\mathcal{H}_0]-0|= |E\tilde{S}_1|<\chi$,
as $S_0=0$ and $\mathcal{H}_0$ is trivial. On the other hand,
$\chi>E\sup_{t\in [0,1]} |E[\tilde{S}_1\vert
\mathcal{F}_t]-S_t|\geq E|\tilde{S}_1
-S_1|$. Noting that $ES_1=1/2$,
this would mean $E\tilde{S}_1>1/4$ while we have just seen that $E\tilde{S}_1<1/4$,
a contradiction.}
\end{example}

\begin{corollary}\label{added} Let $\chi>0$ be any fixed number. Let $S$ be a c\`adl\`ag
process which is sticky with respect to $\f$. Let Assumption
\ref{large} be in force. Let $\mathcal{H}_t=\FF_t\vee \mathcal{G}_t$ for each $t\in [0, T]$. 
For each $p\geq 1$ 
there exists $Q\sim P$ and a 
$d$-dimensional $Q$-martingale $\tilde{S}$ (with respect to $\mathbb{H}$) such that 
\begin{equation}
E_Q\sup_{t\in [0,T]}|S_t-\tilde{S}_t|^p<\chi.
\end{equation}
\end{corollary}
\begin{proof}
Indeed, let $g(x):=x^p$, $x\geq 0$, and apply Theorem \ref{main1}. 
\end{proof}

\begin{remark}
{\rm In the case where $S$ is a continuous process, Theorem \ref{main1} was 
proved in Bender et al. \cite{bps} in a slightly different form. In that paper $S$ is assumed to be
positive and $\tilde{S}$ is shown to satisfy 
\begin{equation}\label{rofi}
\sup_{t\in [0,T]}|S_t/\tilde{S}_t-1|<\chi\mbox{ a.s.}
\end{equation}
Minor modifications of that argument would work for not necessarily positive, continuous $S$ and
they would lead to \eqref{reffi} instead of \eqref{rofi}, \emph{without} using Assumption \ref{large}. 
Thus the novelty of Theorem \ref{main1}
lies in treating the case of discontinuous processes, at the price of requiring Assumption \ref{large}.
We do not know whether this assumption could be dropped.}
\end{remark}

The following lemma will be a key ingredient for the proof of Theorem \ref{main1}.
We now consider a discrete-time filtration $(\mathcal{K}_n)_{n\in\mathbb{N}}$. 
We introduce some notation that will be used in the sequel.
For an $\mathbb{R}^d$-valued random variable $X$, let $\mathcal{D}(X)$ be the smallest affine
subspace containing the support
of $\mathrm{Law}(X)$.
Let $\mathcal{S}(X)$ be the relative interior of the convex hull of the support of $\mathrm{Law}(X)$.
The meanings of $\mathcal{D}(\mu)$,
$\mathcal{S}(\mu)$ are analogous for a probability $\mu$ on $\mathbb{R}^d$. 
We denote by $B(x,r)$ the closed ball of radius $r\geq 0$ around $x\in\mathbb{R}^d$.

\begin{lemma}\label{infamous} Fix any $\varepsilon>0$ and assume that $w:\mathbb{R}^d\to\mathbb{R}_+$ 
is a continuous
function with $w(0)=0$ and $w(x)\geq |x|$. 
Let $(M_n)_{n\in\mathbb{N}}$ be a discrete-time 
process adapted to $(\mathcal{K}_n)_{n\in\mathbb{N}}$.
Assume that $0\in\mathcal{S}(Q_n(\cdot,\omega))$
a.s. and, for all $\epsilon>0$, $Q_n(B(0,\epsilon),\omega)>0$ a.s., where $Q_n(\cdot,\cdot)$ is
the conditional law of $M_n-M_{n-1}$ with respect to $\mathcal{K}_{n-1}$, $n\geq 1$.
Assume that there exists a random variable $M_{\infty}$ and $A_n\in\mathcal{K}_n$ such that $1_{A_n}$ 
increases to $1$ a.s.
when $n\to\infty$ 
and 
\begin{equation}\label{masni}
\{M_k=M_{\infty},\ k\geq n\}\supset A_n\supset \{|M_n-M_{n-1}|<\varepsilon\}
\end{equation}
for all $n$. 
Then there is a $Q\sim P$ such that $M_n$, $n\in \mathbb{N}\cup\{\infty\}$, is a uniformly integrable
$Q$-martingale  and 
\begin{equation}\label{obelix}
E_Q\left[\sum_{n=1}^{\infty} w(M_n-M_{n-1})\right]<\varepsilon.
\end{equation}
\end{lemma}
\begin{proof}
By applying  Lemma \ref{lema}  with the choices 
\[
X:=M_n-M_{n-1},\ \mathcal{K}:=\mathcal{K}_{n-1},\ 
\eta:=\varepsilon/2^{n},
\]  
 we obtain $j_n(y,\omega)$ for  each $n\geq 1$.  Define 
 \[
Z_n(\omega):=j_n(M_n(\omega)-M_{n-1}(\omega),\omega).
\]
Set $dQ/dP:=\prod_{n=1}^{\infty} Z_n$. Note that, by the last statement of
Lemma \ref{lema} and by \eqref{masni}, we have that $Z_k=1$ for all $k\geq n+1$ on $A_n$. Hence, for almost all $\omega$, only
finitely many $Z_n(\omega)$ differ from $1$. So the infinite product converges almost surely. We claim that $Q(\Omega)=1$. Indeed, 
by monotone convergence, we have
\begin{eqnarray*}
E\frac{dQ}{dP}&=&\lim_{n\to\infty}E\left[1_{A_n}\frac{dQ}{dP}\right]=\lim_{n\to\infty}
E\left[1_{A_n}Z_{n}\cdots Z_1\right]\\
&\geq& 1-\limsup_{n\to\infty}E\left[1_{A_n^C}Z_{n}\cdots Z_1\right].\\
\end{eqnarray*}
By \eqref{masni}, $A_n^C\subset \{|M_n-M_{n-1}|\geq\varepsilon\}$,  and by Lemma \ref{lema} we have
$$
E\left[Z_{n}1_{\{|M_n-M_{n-1}|\geq\varepsilon\}}\vert\mathcal{K}_{n-1}\right]<\varepsilon/2^{n}.
$$
It follows that
\begin{eqnarray*}
E\left[1_{A_n^C}Z_{n}\cdots Z_1\right] &=& E\left[E[1_{A_n^C}Z_{n}|\mathcal{K}_{n-1}]Z_{n-1}\cdots Z_1\right]\\
&\leq& (\varepsilon/2^{n}) E\left[Z_{n-1}\cdots Z_1\right]
 = \varepsilon/2^{n}\to 0,
\end{eqnarray*}
as $\ n \to \infty$, showing that $Q(\Omega)\geq 1$. Fatou's lemma ensures $Q(\Omega)\leq 1$. 

Now it remains to show that $M_n$ is a uniformly integrable martingale under $Q$. 
The martingale property of $M_n$, $n\in\mathbb{N}$ under $Q$ is clear from the construction of $Q$. 
Since $w(x)\geq |x|$, \eqref{obelix} implies that $M_n$ converges to $M_{\infty}$ 
in $L^1(Q)$ hence $M_n$, $n\in\mathbb{N}\cup\{\infty\}$ is a uniformly integrable martingale under $Q$.
\end{proof}

\begin{remark}\label{neige}
{\rm Assume that $w(x)\geq |x|^{\kappa}$, $x\in\mathbb{R}^d$ with some $\kappa\geq 1$. Then 
a trivial modification of the proof of Lemma \ref{infamous} yields not only \eqref{obelix} but also
\[
\sum_{n=1}^{\infty} E_Q^{1/\kappa} |M_n-M_{n-1}|^{\kappa}<\varepsilon,
\]
which implies
\[
E_Q^{1/\kappa} [\sup_n |M_n|^{\kappa}]<\infty,
\]
whenever $E|M_0|^{\kappa}<\infty$, in particular, when $M_0$ is constant. 
} 
\end{remark}

%Let $S_t$, $t\in [0,T]$ be a $d$-dimensional c\`adl\`ag adapted process on a stochastic basis
%$(\Omega,\mathcal{F},(\mathcal{F}_t)_{t\in [0,T]},P)$.
%We prove the result of Theorem \ref{main1} under an additional assumption first. 

\begin{proposition}\label{dublin}  Assume that $S$ is sticky with respect to $\mathbb{F}$. Let 
$g:\mathbb{R}_+\to\mathbb{R}_+$ be any convex function with $g(0)=0$. Assume that for the sequence $(S_{\tau_n})_{n\geq 0}$
we have $0\in\mathcal{S}(P(S_{\tau_{n+1}}-
S_{\tau_n}\in\cdot|\mathcal{F}_{\tau_n}))$ almost surely,  where
the stopping times $\tau_n$ are recursively defined by 
$$
\tau_0=0,\quad \tau_{n+1}:=\inf\{t>\tau_{n}:\, |S_t-S_{\tau_n}|\geq \varepsilon \}\wedge T, 
$$ 
for some  $\varepsilon>0$.
Then there exists $Q\sim P$ and a 
$d$-dimensional $Q$-martingale $\tilde{S}$ with respect to the filtration $\mathbb{F}$ such that $\tilde{S}_0=S_0$, $\tilde{S}_T=S_T$, and 
\[
E_Qg(\sup_{t\in [0,T]}|S_t-\tilde{S}_t|)<g(2\varepsilon)+2\sqrt{\varepsilon},
\] 
where the latter expression can be made arbitrarily small when $\varepsilon\to 0$.
If $S$ has continuous trajectories then even 
\begin{equation}\label{mala}
\sup_{t\in [0,T]}|S_t-\tilde{S}_t|<2{\varepsilon}
\end{equation}
holds almost surely.  If $S$ is (strictly) positive then
so is $\tilde{S}$. 
\end{proposition}

\begin{remark}\label{larmes}{\rm Comparing Proposition \ref{dublin} to Theorem \ref{main1}, the former does not
require Assumption \ref{large} and it provides $\tilde{S}$ satisfying $S_T=\tilde{S}_T$ but this
comes at the price of a hypothesis involving the $\tau_n$. Still, Proposition \ref{dublin}
improves on previous results \emph{even} in the case of continuous $S$. Indeed, if $S$ has the
CFS property then the conditions of Proposition \ref{dublin} are easily seen to hold, by an argument
similar to Lemma A.2 of Guasoni et al. \cite{grs}. Therefore Proposition 
\ref{dublin} strengthens the conclusion of Theorem 2.11 of Guasoni 
et al. \cite{grs} (see also Theorem 2.1 of the same paper): we get $\tilde{S}$ as in \eqref{mala} but 
satisfying $S_0=\tilde{S}_0$
and $S_T=\tilde{S}_T$ a.s. as well.}
\end{remark}
%We should mention that the conditions of Proposition \ref{dublin} are more stringent than those
%of Theorem \ref{infamous}. However the conclusion is  stronger: here we can guarantee that
%$S_T=\tilde{S}_T$  almost surely.

\begin{proof}[Proof of Proposition \ref{dublin}] 
The idea here is to apply Lemma \ref{infamous} to $S$ sampled at the stopping times $\tau_n$. 
The $Q$ constructed is such that all the increments $S_{\tau_n}-S_{\tau_{n-1}}$ will be ``small''
but $S_{\tau_n}=M_n$, $n\in\mathbb{N}$ is a $Q$-martingale. Then
$\tilde{S}$ will be just the continuous-time $Q$-martingale with terminal value $S_T=M_{\infty}$ and,
by the choice of $Q$, $\sup_{t\in [0,T]} |S_t-\tilde{S}_t|$ will also be ``small''.

Note that $g$ is necessarily continuous (even at $0$).
Set $M_n:=S_{\tau_n}$  and  $\mathcal{K}_n:=\mathcal{F}_{\tau_n}$ for all $n\in\mathbb{N}$. 
Using the notations of Lemma \ref{infamous}, the conditions of the present proposition imply that $0\in\mathcal{S}(Q_n(\cdot,\omega))$
almost surely. The stickiness property guarantees that, for any small real number $\zeta>0$ and  all 
$n\geq 1$, $Q_n(B(0,\zeta),\omega)>0$ almost surely.
Define $A_n:=\{\tau_n=T\}\in\mathcal{K}_n$. As $S$ has c\`adl\`ag paths, for almost all $\omega$, 
the increasing sequence  $\tau_n(\omega)$ can not have a limit strictly less than $T$. This shows that $\tau_n(\omega)=T$ for all $n\geq m(\omega)$
for some $m(\omega)\in\mathbb{N}$ almost surely. Therefore $1_{A_n}$ increases to $1$  almost surely.  

Set $M_{\infty}:=S_T$. From the definition of $\tau_n$ we have 
$\{|M_n-M_{n-1}|<\varepsilon\}\subset A_n$ and therefore \eqref{masni} holds. Using Lemma 
\ref{infamous} with the choice $w(x):=g^2(2|x|)+|x|$ we obtain $Q$. Now define 
$\tilde{S}_t:=E_Q[S_T|\mathcal{F}_t]$,
$t\in [0,T]$ (we take a c\`adl\`ag version of this $Q$-martingale). This definition makes sense
since $S_T$ is $Q$-integrable by $|x|\leq w(x)$ and \eqref{obelix}.
We clearly have $\tilde{S}_0=S_0$ and $\tilde{S}_T=S_T$. It remains to estimate 
\[
\sup_{t\in [0,T]}|S_t-\tilde{S}_t|.
\] 
Fix $t,n$ for a moment and let us work on the event $B_n:=\{\tau_n\leq t<\tau_{n+1}\}$
till further notice.
We have
\begin{eqnarray*}
|S_t-\tilde{S}_t|=|S_{t\wedge \tau_{n+1}}-\tilde{S}_{t\wedge\tau_{n+1}}|=
|S_{t\wedge \tau_{n+1}}-E_Q[\tilde{S}_{\tau_{n+1}}|\mathcal{F}_{t\wedge\tau_{n+1}}]|,
\end{eqnarray*}
by the $Q$-martingale property of $\tilde{S}$. We further have
\begin{eqnarray}\nonumber
\left|E_Q\left[S_{t\wedge \tau_{n+1}}-\tilde{S}_{\tau_{n+1}}\Big|\mathcal{F}_{t\wedge\tau_{n+1}}\right]\right| 
&\leq& E_Q\left[|S_{t\wedge \tau_{n+1}}-\tilde{S}_{\tau_{n}}|+|S_{\tau_n}-\tilde{S}_{\tau_{n+1}}|
\Big|\mathcal{F}_{t\wedge\tau_{n+1}}\right]\\
&\leq & E_Q\left[\varepsilon +|M_{n+1}-M_n|
\Big|\mathcal{F}_{t\wedge\tau_{n+1}}\right],\label{ttt}
\end{eqnarray}
which follows from the definitions of $B_n$,  $\tau_n$,  and  
$$
\tilde{S}_{\tau_{k}}=E_Q[S_T\vert\mathcal{F}_{\tau_k}]=M_{k}=S_{\tau_{k}}
$$
for both $k=n$ and $k=n+1$. Hence we get

\begin{eqnarray*}
g(|S_t-\tilde{S}_t|) &\leq& g\left(\varepsilon + E_Q\left[|M_{n+1}-M_n|\Big\vert\mathcal{F}_{t\wedge\tau_{n+1}}
\right]\right) \\ 
&\leq& \frac{1}{2}\left(g(2\varepsilon) + g\left(E_Q\left[2|M_{n+1}-M_n|\Big\vert
\mathcal{F}_{t\wedge\tau_{n+1}}\right]\right)
\right) \\
&\leq&  g(2\varepsilon)+E_Q[g(2|M_{n+1}-M_n|)\vert\mathcal{F}_{t\wedge\tau_{n+1}}],\\
\end{eqnarray*}
by the convexity of $g$. Noting that $g$ is necessarily non-decreasing, we get

\begin{eqnarray*}
g(|S_t-\tilde{S}_t|) &\leq& g(2\varepsilon)+E_Q[g(2\sup_n|M_{n+1}-M_n|)\vert\mathcal{F}_{t\wedge\tau_{n+1}}]\\
&\leq& g(2\varepsilon)+E_Q[L_T\vert\mathcal{F}_{t\wedge\tau_{n+1}}]\\
&\leq& g(2\varepsilon)+\sup_{s\in [0,T]}L_s, 
\end{eqnarray*}
for the positive $Q$-martingale $L_s:=E_Q[g(2\sup_n|M_{n+1}-M_n|)\vert\mathcal{F}_s]$,
$s\in [0,T]$. The right-hand side here, however, does not depend either on $t$ or on $n$
so this estimate, in fact, holds a.s. on $\Omega=\cup_n B_n$.
Hence
\begin{eqnarray*}
E_Q g(\sup_{t\in [0,T]}|S_t-\tilde{S}_t|) &\leq& g(2\varepsilon)+ E_Q\left[\sup_{s\in [0,T]}L_s\right]\\
&\leq& g(2\varepsilon)+E_Q^{1/2}\left[\sup_{s\in [0,T]}L^2_s\right]\\
&\leq& g(2\varepsilon)+2E_Q^{1/2}L_T^2\\
&\leq& g(2\varepsilon) +2E_Q^{1/2}\left[\sum_{n=0}^{\infty}w(M_{n+1}-M_n)\right]\\
&\leq& g(2\varepsilon)+2\sqrt{\varepsilon},\\
\end{eqnarray*}
using Doob's inequality and Lemma \ref{infamous}. Positivity of $\tilde{S}$ is clear
since $\tilde{S}_t=E_Q[S_T|\mathcal{F}_t]$ and $S_T$ is positive. If $S$ is continuous then
$|S_{\tau_n}-S_{\tau_{n-1}}|\leq \varepsilon$ for all $n$, so we can deduce \eqref{mala} directly from \eqref{ttt}.
\end{proof}

\begin{lemma} \label{sticky-lem} Let $X$ and $Y$ be two independent c\`adl\`ag processes. 
Let $\mathbb{F}=(\FF_t)_{t\in [0, T]}$ and $\mathbb{G}=(\mathcal{G}_t)_{t\in [0, T]}$ be independent, complete, right-continuous filtrations to which $X$ and $Y$ are adapted, respectively. Let 
$\mathcal{H}_t=\FF_t \vee \mathcal{G}_t$ for all $t\in [0, T]$. Then $(\mathcal{H}_t)_{t\in [0, T]}$ is a 
complete, right-continuous filtration. If $X$ is  sticky with respect to $\mathbb{F}$ and $Y$ is sticky with 
respect to $\mathbb{G}$, then all of $X, Y, X\pm Y$ and $(X,Y)$ are sticky with respect to the filtration 
$\mathbb{H}=(\mathcal{H}_t)_{t\in [0, T]}$. 
\end{lemma}
\begin{proof}  First observe that $\mathbb{H}$ is a complete filtration as both 
$\mathbb{F}$ and $\mathbb{G}$ are complete. Therefore it is enough to prove that 
$\mathbb{H}$ is right-continuous and, to this end, it is enough to prove 
$E[Z|\mathcal{H}_t]=E[Z|\mathcal{H}_{t+}]$ for any $\mathcal{H}_T$-measurable nonnegative random variable 
$Z$ and for all $0\le t<T$. By the monotone class theorem, it is enough to prove this equality for 
$Z=UV$ where $U\geq 0$ is $\mathcal{F}_T$-measurable and $V\geq 0$ is $\mathcal{G}_T$-measurable.
However, Lemma \ref{egyszeru} implies that
\begin{eqnarray*}
\lim_{h\to 0} E[UV\vert\mathcal{H}_{t+h}] &=& \lim_{h\to 0} E[U\vert\mathcal{F}_{t+h}] 
E[V\vert\mathcal{G}_{t+h}]\\ &=& E[U\vert\mathcal{F}_{t}] 
E[V\vert\mathcal{G}_{t}]=E[UV\vert\mathcal{H}_{t}],
\end{eqnarray*}
by the right-continuity of $\mathcal{F}_t$, $\mathcal{G}_t$, $t\in [0,T]$. This shows right-continuity
of $\mathcal{H}_t$, $t\in [0,T]$.

To show the second claim in the Lemma it is sufficient to show that $X_t, Y_t$ are sticky for $\mathbb{H}$. 
The stickiness of $X_t\pm Y_t$ with respect to $\mathbb{H}$ then follows from Proposition 1 of 
Sayit and Viens \cite{sv} (continuous functions of sticky processes are sticky). We only show that 
$X$ is sticky for $\mathbb{H}$, the argument for $Y$ being identical. Since $X_t$ is a right-continuous process we need to check
\[
P(\sup_{t\in [s, T]}|X_t-X_{s}|<\kappa|\mathcal{H}_s)>0 \mbox{ a.s.},
\]
for any $\kappa>0$ and any deterministic $s\in [0, T]$ (see Lemma \ref{stickiness2} above). This follows by Lemma \ref{egyszeru} 
from
\[
P(\sup_{t\in [s, T]}|X_t-X_{s}|<\e|\mathcal{F}_s\vee \mathcal{G}_s)=P(\sup_{t\in [s, T]}|X_t-X_{s}|<\e |\mathcal{F}_s)> 0 \mbox{ a.s.,}
\]
as $\mathcal{F}_s\vee\sigma(X)$ is independent from $\mathcal{G}_s$ and $X$ is sticky for $\mathbb{F}$. 

To see the last statement, apply Lemma \ref{egyszeru} to obtain
\begin{eqnarray*}
P(\sup_{t\in [s, T]}|X_t-X_{s}|<\kappa,\,\sup_{t\in [s, T]}|Y_t-Y_{s}|<\kappa|\mathcal{H}_s) &=&\\
P(\sup_{t\in [s, T]}|X_t-X_{s}|<\kappa|\mathcal{F}_s)P(\sup_{t\in [s, T]}|Y_t-Y_{s}|<\kappa|\mathcal{G}_s) &>& 0,
\end{eqnarray*}
by the stickiness of $X,Y$ with their respective filtrations. 
\end{proof}

Now, using the previous arguments, it is possible to establish Theorem \ref{main1}, too. Before presenting the proof 
we make some important observations.

\begin{remark} \label{newremark}{\rm
Let $B_t$ be a Brownian motion with respect to a filtration $\mathcal{L}_t$ and let $0\leq\theta< T$ be an arbitrary deterministic time.
Then, by Theorem 6.1 in Chapter 2 of Karatzas and Shreve \cite{kash}, $B_{s+\theta}-B_{\theta}$, $s\geq 0$, is a Brownian motion independent of
$\mathcal{L}_{\theta}$. Let $C_0[\theta,T]$ denote the space of $\mathbb{R}^d$-valued continuous functions on
$[\theta,T]$ which are $0$ at $\theta$.

Let us first note that the mapping $f\to \sup_{s\in [\theta,T]} |B_s(\omega)-B_{\theta}(\omega)-f_s|$,
$f\in C_0[\theta,T]$, is continuous for a.e. $\omega\in\Omega$ and hence it is jointly measurable in $(\omega,f)$.
It follows that $\{\sup_{s\in [\theta,T]} |B_s(\omega)-B_{\theta}(\omega)-G_s|<\epsilon\}$ and $\{\sup_{s\in [\theta,T]} |B_s(\omega)-B_{\theta}(\omega)-G_s|\leq\epsilon\}$ are events for each $\epsilon>0$, where $G$ is a random element of $C_0[\theta,T]$.

Now define $q(\epsilon,f):=P(\sup_{s\in [\theta,T]} |B_s-B_{\theta}-f_s|<\epsilon)$ and notice that $q(\epsilon,f)>0$
for all $f\in C_0[\theta,T]$ as $B_t-B_{\theta}$, $t\in [\theta,T]$, has full support on $C_0[\theta,T]$.
Fatou's lemma for events shows that $f\to q(\epsilon,f)$ is lower semicontinuous. 
Notice that, if $G^n$ are $\mathcal{L}_{\theta}$-measurable $C_0[\theta,T]$-valued random variables taking
only countable many values, then
\[
P(\sup_{s\in [\theta,T]} |B_s-B_{\theta}-G^n_s|<\epsilon\vert\mathcal{L}_{\theta})=q(\epsilon,G^n).
\]
Now let $G$ be an arbitrary $\mathcal{L}_{\theta}$-measurable random element in $C_0[\theta,T]$.
Choose a sequence $G^n$, $n\in\mathbb{N}$,  of discrete $\mathcal{L}_{\theta}$-measurable random elements in $C_0[\theta,T]$ 
such that $G^n$ tend to $G$ almost surely.
Lower semicontinuity of $q(\epsilon,\cdot)$ and Fatou's lemma for events imply that
\begin{eqnarray}\nonumber
0 &<& q(\epsilon,G)\leq\liminf_n q(\epsilon,G^n)=\liminf_n 
P(\sup_{s\in [\theta,T]} |B_s-B_{\theta}-G^n_s|<\epsilon\vert\mathcal{L}_{\theta})\\
\nonumber &\leq& \liminf_n P(\sup_{s\in [\theta,T]} |B_s-B_{\theta}-G^n_s|\leq\epsilon\vert\mathcal{L}_{\theta})\\
\nonumber &\leq& \limsup_n
P(\sup_{s\in [\theta,T]} |B_s-B_{\theta}-G^n_s|\leq\epsilon\vert\mathcal{L}_{\theta})\\
&\leq& P(\sup_{s\in [\theta,T]} |B_s-B_{\theta}-G_s|\leq \epsilon\vert\mathcal{L}_{\theta}).\label{karc}
\end{eqnarray}
}
\end{remark}

\begin{proof}[Proof of Theorem \ref{main1}.] We wish to apply Proposition \ref{dublin} but
$S$ does not necessarily satisfy $0\in\mathcal{S}(P(S_{\tau_{n+1}}-
S_{\tau_n}\in\cdot|\mathcal{F}_{\tau_n}))$. To fix this, we perturb $S$ by an independent ``small noise''
$W$ such that, for $Y_t:=S_t+W_t$, $0\in\mathcal{S}(P(Y_{\tau_{n+1}}-
Y_{\tau_n}\in\cdot|\mathcal{F}_{\tau_n}))$ holds. As $Y$ is close to $S$, the $\tilde{S}$ constructed
for $Y$ by Proposition \ref{dublin} will also be close to $S$.

Fix any $\varepsilon >0$. Let $B_t=(B_t^1, B_t^2, \cdots, B_t^d)$ be the Brownian motion
of Assumption \ref{large}. We remark that $B$ clearly has the CFS 
property. Let $\pi: (-\infty, +\infty)\rightarrow (-\varepsilon, +\varepsilon)$ 
be a bijective and Lipschitz-continuous (deterministic) function. Let 
$F:\mathbb{R}^d\to (-\varepsilon, +\varepsilon)^d$ be defined by $F(x^1,\ldots,x^d):=(\pi(x^1),\ldots,\pi(x^d))$.
Denote by $L$ a Lipschitz constant for the mapping $F$. Now set
$$
W_t:=(W_t^1, W_t^2, \cdots, W_t^d):=F(B_t^1, B_t^2, \cdots, B_t^d).
$$
Define $Y_t=S_t+W_t$. From Lemma \ref{sticky-lem} above,  $Y$ is sticky for the filtration $\mathbb{H}$. 

%We claim that $Z$ has the CFS property in $(-\e, \e)^d$ with respect to the filtration $\mathbb{H}$.
%To see this,  note that the CFS property of $B$ with respect to $\mathbb{G}$ is equivalent to the stickiness of 
%$B+h$ with respect to $\mathbb{G}$ for any continuous function $h:[0,T]\to\mathbb{R}^d$ with $h(0)=0$. 
%From Lemma \ref{sticky-lem},  $B+h$ is sticky for $\mathbb{H}$ for any such $h$ which, in turn, 
%implies that $B$ has CFS for $\mathbb{H}$. Then the CFS property of $Z$ in $(-\e, \e)^d$ with respect to 
%$\mathbb{H}$  is clear as $f$ is a strictly monotone, continuous function. 

%We remark that the argument of 
%Lemma A.1 of Guasoni et al. \cite{grs} implies that the CFS-$(-\e, \e)^d$ property holds for $Z$
%at each stopping time, too (we refrain from the precise formulation of this statement).

%Below we use the norm $|x|=\max \{|x_1|, |x_2|, \cdots, |x_d|\}$. 

For each positive integer $n\geq 1$, define
\[
\tau_n=\inf\{t\geq \tau_{n-1}: \; |Y_t-Y_{\tau_{n-1}}|\geq \varepsilon\}, \; \; \; \tau_0=0.
\]
These are stopping times with respect to the filtration $\mathbb{H}$. We would like to show  that 
$\triangle_n:=Y_{\tau_n}-Y_{\tau_{n-1}}$ satisfies  
\begin{eqnarray}\label{one-1}
0\in \mathcal{S}(P(\triangle_n\in\cdot |\mathcal{H}_{\tau_{n-1}}))
\end{eqnarray}
almost surely, for each $n$. Fixing $n$, from now on we are working on the set $\{\tau_{n-1}<T\}$ (since \eqref{one-1}
is trivial on $\{\tau_{n-1}=T\}$). We will write $\tau:=\tau_{n-1}$ henceforth.

Let $0<\eta<\varepsilon/2$ be an $\mathcal{H}_{\tau}$-measurable random variable
such that $B(W_{\tau},\eta)\subset (-\varepsilon,\varepsilon)^d$. 
Working separately on events of the form $\{\eta\geq 1/j\}$, $j\in\mathbb{N}$ we
may and will assume that $\eta$ is a constant.

Fix $x\in B(0,\eta/2)\cap\mathbb{Q}^d$. It suffices to show that 
\begin{equation}\label{tor}
x\in\mathrm{supp}\, P(\triangle_n\in\cdot |\mathcal{H}_{\tau})\mbox{ a.s.}
\end{equation}
on $\{\tau<T\}$ since this implies that, outside a null set of $\omega$'s, 
$\mathrm{supp}\, P(\triangle_n\in\cdot |\mathcal{H}_{\tau})(\omega)$
contains $B(0,\eta/2)\cap\mathbb{Q}^d$ hence, being a closed set,
also the whole of $B(0,\eta/2)$. The statement \eqref{tor} will follow if, for each 
$l\in\mathbb{N}$, 
\begin{eqnarray}\nonumber
P(\triangle_n\in B(x,(1+\eta)/l)|\mathcal{H}_{\tau}) &\geq&\\
\label{marci} P(\tau_n=T,\ \triangle_n\in B(x,(1+\eta)/l)|\mathcal{H}_{\tau}) &>& 0
\end{eqnarray}
almost surely on $\{\tau<T\}$.

Fix $l\in\mathbb{N}$. We will now prove \eqref{marci}. 
To this end, define 
\[
J_t=B_{\tau}(\omega)-F^{-1}\left(\frac{t-\tau}{T-\tau}x+F(B_{\tau}(\omega))\right),\mbox{ for }t\in [\tau(\omega), T],
\]
and $J_t=0$, $t<\tau(\omega)$.
This definition makes sense since $|\frac{t-\tau}{T-\tau}x|\leq \eta/2$ and, by the choice of $\eta$,
$\frac{t-\tau}{T-\tau}x+F(B_{\tau})\in (-\varepsilon,\varepsilon)^d$. 
Furthermore, define the events
\begin{eqnarray*}
D(l) &:=& \{\sup_{t\in [\tau,T]} |S_t-S_{\tau}|<\eta/l\},\\ 
K(l) &:=& \{
\sup_{t\in [\tau,T]} |W_t-W_{\tau}|<\eta,\ |W_T-W_{\tau}-x|\in B(0,1/l)\},\\
H(l) &:=& \left\{ \sup_{t\in [\tau,T]} \left|B_t-B_{\tau}+J_t\right|
< \frac{1}{L}\min\{1/l,\eta/2\}\right\}.
\end{eqnarray*}

%Clearly, $\tau_n=T$ and $\triangle_n\in B(x,(1+\eta)/l)$ hold on $D(l)\cap H(l)$. 
%$P(D(l)\cap H(l)|\mathcal{H}_{\tau})>0$ a.s. for all $l$ would imply \eqref{tor} since
%if $x$ were outside the support in question then a small ball around it would also
%be outside that support. 

We first show that 
\begin{equation}\label{hora}
P(D(l)\cap H(l)|\mathcal{H}_{\tau})>0\mbox{ a.s.} 
\end{equation}
on $\{\tau<T\}$. For this, it is sufficient to show that for any $A\in \mathcal{H}_{\tau}$ $A\subset \{\tau<T\}$ and $P(A)>0$, the relation $P(A\cap D(l)\cap H(l))>0$ holds. 
Fix such an $A$ and a deterministic number 
$0<\epsilon_0< \min \{\eta/l,  \frac{1}{L}\min\{1/l,\eta/2\}\}$.   

As $J_t$ is an $\mathcal{H}_{\tau}$-measurable continuous process with $J_{\tau}=0$,
there exists a deterministic number $\theta \le T$ such that  the event $A_1=A\cap \{\sup_{t\in [\tau, \theta]}|J_t|\le \frac{ \epsilon_0}{6}\}\cap \{\tau <\theta\}$ 
has positive probability. Note that $A_1\in\mathcal{H}_{\theta}\cap \mathcal{H}_{\tau}$. The joint stickiness of the process $(S_t, B_t)$, see Lemma \ref{sticky-lem},
shows that the event 
\[
A_2=A_1\cap \{\sup_{t\in [\tau, \theta]}|S_t-S_{\tau}|\le \frac{\epsilon_0}{2}\}\cap \{\sup_{t\in [\tau, \theta]}|B_t-B_{\tau}|\le \frac{\epsilon_0}{6}\}
\]
has positive probability. Now observe that $A_2\cap d(l)\subset D(l)$ where
\[
d(l)=:\{\sup_{t\in [\theta, T]}|S_t-S_{\theta}| \le \frac{\epsilon_0}{2}\}.
\]
We also claim $A_2\cap h(l)\subset H(l)$, where
\[
h(l)=:\{\sup_{t\in [\theta, T]} |B_t-B_{\theta}+J_t-J_{\theta}|\le \frac{\epsilon_0}{3}\}.
\]
Indeed,
\begin{eqnarray*}
\sup_{t\in [\tau,T]} |B_t-B_{\tau}+J_t| &\leq& 
\sup_{t\in [\tau,\theta]} |B_t-B_{\tau}+J_t| + \sup_{t\in [\theta,T]} |B_t-B_{\tau}+J_t|\\
&\leq&
\sup_{t\in [\tau,\theta]} |B_t-B_{\tau}| + \sup_{t\in [\tau,\theta]}|J_t| +
\sup_{t\in [\theta,T]} |B_t-B_{\theta}+J_t-J_{\theta}|\\
&+& |B_{\theta}+J_{\theta}-B_{\tau}|\\ &<&
\epsilon_0/6 +\epsilon_0/6+\epsilon_0/3 + |B_{\theta}-B_{\tau}| +|J(\theta)|\\ &\leq &\epsilon_0
\end{eqnarray*}
on $A_2\cap h(l)$ so $A_2\cap h(l)\subset H(l)$.

We conclude that
\begin{eqnarray}\label{horaa}
A_2\cap d(l)\cap h(l)
\subset A\cap D(l)\cap H(l).
\end{eqnarray}
Therefore it is sufficient to show that the left-hand side of (\ref{horaa}) has positive probability. Since 
$A_2\in\mathcal{H}_{\theta}$, it is sufficient to show that
\begin{eqnarray}\label{horaaa}
P(d(l)\cap h(l)\vert  \mathcal{H_{\theta}})>0\mbox{ a.s.}
\end{eqnarray}

Define $\mathcal{L}_t:=\mathcal{G}_t\vee \mathcal{F}_T$. We can write (\ref{horaaa}) as follows:
\begin{eqnarray*}
\begin{array}{ll}
P(d(l)\cap h(l)\vert\mathcal{H}_{\theta})&=E[E[1_{d(l)\cap h(l)}\vert\mathcal{L}_{\theta}]\vert\mathcal{H}_{\theta}]\\
&=E[1_{d(l)}E[1_{h(l)}\vert\mathcal{L}_{\theta}]\vert\mathcal{H}_{\theta}].
\end{array}
\end{eqnarray*}
Notice that $G_s:=J_s-J_{\theta}$, $s\in [\theta,T]$,  is a $\mathcal{H}_{\theta}\subset\mathcal{L}_{\theta}$-measurable random element in $C_0[\theta,T]$ so
$E[1_{h(l)}\vert\mathcal{L}_{\theta}]\geq 
q(\epsilon_0/3,G)>0$ a.s. by \eqref{karc} in Remark \ref{newremark} above. It follows that
\[
P(d(l)\cap h(l)\vert\mathcal{H}_{\theta})\geq q(\epsilon_0/3,G) P(d(l)\vert\mathcal{H}_{\theta})>0,
\]
by the stickiness of $S$ with respect to $\mathbb{F}$ and by Lemma \ref{sticky-lem}. 
We conclude that (\ref{hora}) holds. 

We will now show that
$$
H(l)\cap D(l)\subset K(l)\cap D(l)\subset \{\tau_n=T,\ \triangle_n\in B(x,(1+\eta)/l) \},
$$
which will entail \eqref{marci}, in view of \eqref{hora}.%, \eqref{mack1} and \eqref{mack}.

The second containment is trivial since $\sup_{t\in [\tau,T]} |W_t-W_{\tau}|<\eta$ and 
$\sup_{t\in [\tau,T]}|S_t-S_{\tau}|<\eta/l$ entail $\sup_{t\in [\tau,T]} |Y_t-Y_{\tau}|<\varepsilon$ 
by $\eta<\varepsilon/2$ which implies $\tau_n=T$. Obviously, $|W_T-W_{\tau}-x|\in B(0,1/l)$ together with $\sup_{t\in [\tau,T]} |S_t-S_{\tau}|<\eta/l$
imply $\triangle_n\in B(x,(1+\eta)/l)$. 

For the first containment, the Lipschitz property of $F$ and $W=F(B)$
clearly imply that
$$
\left|W_t-\frac{t-\tau}{T-\tau}x-W_{\tau}\right|\leq L
\left|B_t-F^{-1}\left(\frac{t-\tau}{T-\tau}x+F(B_{\tau})\right)\right|<\min\left\{\frac{1}{l},\frac{\eta}{2}\right\},
$$
on $H(l)$. For $t=T$ this gives $|W_T-W_{\tau}-x|< 1/l$ whereas
$$
\sup_{t\in [\tau,T]} |W_t-W_{\tau}|\leq \sup_{t\in [\tau,T]} \left|W_t-W_{\tau}-\frac{t-\tau}{T-\tau}x\right| +
\sup_{t\in [\tau, T]}\left|\frac{t-\tau}{T-\tau}x\right|<\eta/2+\eta/2=\eta,
$$
showing $H(l)\subset K(l)$.

Now apply Proposition \ref{dublin} to the process $Y$ with the convex function
$x\to g(2x)$ to obtain $\tilde{S}$. We get 
\begin{eqnarray*}
E_Qg\left(\sup_{t\in [0,T]} |S_t-\tilde{S}_t|\right) &\leq&  
\frac{1}{2}
E_Qg\left(2\sup_{t\in [0,T]} |Y_t-\tilde{S}_t|\right)+\frac{1}{2}E_Qg\left(2\sup_{t\in [0,T]} |Y_t-S_t|\right)\\
&\leq& \frac{g(4\varepsilon)+2\sqrt{\varepsilon}}{2} + \frac{1}{2}g(2\varepsilon),\\
\end{eqnarray*}
which can be made smaller than $\chi$ when $\varepsilon\to 0$. This completes the proof.
\end{proof}

\begin{remark}\label{maj} {\rm 
By Remark \ref{larmes}, Proposition \ref{dublin} applies to Example 
\ref{w}. If $b=0$ in Example \ref{w} then even Theorem \ref{tetto} 
below applies and one can
approximate many local martingales with true ones. 

Let us now recall Example \ref{w1}. Since $Y_t$ is a martingale and the inverse function of 
$s_{\alpha}$ is strictly monotone, the process $X_t$ satisfies  
$0\in \mathcal{S}(P(X_{\tau}-X_{\theta}\in \cdot|\mathcal{F}_{\theta}))$ almost surely even for all
stopping times $\tau\geq \theta$.  So Proposition \ref{dublin} applies to the case of skew Brownian motion.

Theorem \ref{main1}
applies to the large class of processes presented in Example \ref{w2}.}
\end{remark}

%\begin{remark}\label{subinterval} {\rm In Theorem \ref{main1} it is enough to assume 
%stickiness for $S$ on
%each subinterval $[t_k,t_{k+1}]$, $k=0,\ldots,N-1$ with $0=t_0<t_1<\cdots<t_N=T$. Indeed,
%define $Y$ as in the proof of Theorem \ref{main1} and consider 
%the stopping times $\tau_n^k$, recursively defined by 
%$$
%\tau_0^k=t_{k-1},\quad \tau_{n+1}^k:=\inf\{t>\tau_{n}^k:\, |Y_t-Y_{\tau_n^k}|\geq \varepsilon \}\wedge t_k,
%$$ 
%for each $k=1,\ldots,N-1$.
%Invoking Proposition \ref{dublin} on each subinterval, we get
%$\tilde{S}_t$, $t\in [0,T]$ as required, satisfying even $\tilde{S}_{t_k}=Y_{t_k}$ for 
%$k=0,\ldots,N$.}
%\end{remark}

%\begin{remark}\label{surprise} {\rm Let $S$ be strictly positive.
%A modification of the construction in Theorem \ref{main1} provides $\tilde{S}$ with
%\[
%E_Qg\left(\sup_{t\in [0,T]}\left|S_t/\tilde{S}_t-1\right|\right)<\chi
%\]
%instead of 
%\[
%E_Qg\left(\sup_{t\in [0,T]}|S_t-\tilde{S}_t|\right)<\chi.
%\] 
%In the case where $S$ is a continuous process, we even have 
%\[
%1-\chi\leq \inf_{t\in [0,T]} S_t/\tilde{S}_t\leq \sup_{t\in [0,T]} S_t/\tilde{S}_t\leq 1+\chi,
%\]
%almost surely.}
%\end{remark}

\begin{remark}\label{novum} {\rm At first sight, the argument for Proposition \ref{dublin} looks just a 
variant of 
that of Theorem 1.2 in Guasoni et al.  \cite{grs}, see also Kabanov and Stricker \cite{ks} and 
Subsection 3.6.8 in Kabanov and Safarian \cite{kabanov}. Fine details, however, do 
differ significantly. Not only does Proposition \ref{dublin} cover the
case of processes with jumps, too, but it is sharper even in the case of continuous processes,
as we have already pointed out in Remark \ref{larmes}.}
\end{remark}

%\begin{remark} {\rm In the case where $B_t$, $t\in [0,1]$ is fractional Brownian motion, Theorem 5.3 in \cite{cs} 
%established, for $\alpha>0$ small enough, the existence of $Q\sim P$ and an It\^o process 
%$X_t$ such that $e^{X_t}$ is a local
%martingale under $Q$, $B_t-\alpha\leq X_t\leq B_t$ a.s. for all $t$ and, for all $\varepsilon>0$, $X_t$
%touches both $B_t-\alpha$ and $B_t$ with probability at least $1-\varepsilon$.

%In view of Proposition \ref{dublin} and Remark \ref{surprise} this comes as no surprise. $B$ has the conditional 
%full support property (see Section \ref{exa} below) which is inherited by $S_t:=B_t-\alpha t$, $t\in [0,1]$, for
%any $\alpha>0$. This implies not only stickiness of $S$, but also $0\in\mathcal{S}(P(S_{\tau}-
%S_{\sigma}\in\cdot|\mathcal{F}_{\sigma}))$, by an argument analogous to Lemma A.2 of \cite{grs}.
%Now from Proposition \ref{dublin} we can get $Q\sim P$ and a 
%$Q$-{martingale} $\tilde{S}$ such that $X:=\ln(\tilde{S})$ satisfies $X_0=B_0$,
%$X_1=B_1-\alpha$ {almost surely}.
%}
%\end{remark}

\section{Local martingales}\label{loca}

We denote by $\Vert\cdot\Vert_{tv}$ the total variation norm for finite signed measures
on $(\Omega,\mathcal{F})$. 

\begin{theorem}\label{tetto}
Let $g:\mathbb{R}_+\to\mathbb{R}_+$ be convex with $g(0)=0$ and let $\chi>0$.
Let Assumption \ref{large} be in force.
Assume that $S$ is a sticky local martingale. Then there exists $Q\sim P$ 
with $\Vert Q-P\Vert_{tv}<\chi$ and a 
$d$-dimensional $Q$-martingale $\tilde{S}$ with respect to the enlarged filtration $\mathbb{H}$
such that $E_Qg(\sup_{t\in [0,T]}|S_t-\tilde{S}_t|)<\chi$. If $S$ has continuous trajectories
then even $\sup_{t\in [0,T]}|S_t-\tilde{S}_t|<\chi$ holds a.s. 
Finally, $S$ remains a local martingale with respect to $\mathbb{H}$, too.
\end{theorem}
\begin{proof} 
Let $\sigma_n$, $n\in\mathbb{N}$ be the stopping times increasing to $\infty$ such that $S_{t\wedge\sigma_n}$, $t\in [0,T]$
is a martingale for each $n$. Fix $k$ for the moment.
We will apply the proof of Theorem \ref{main1} (which relies on Proposition
\ref{dublin}), starting from $\sigma_k\wedge T$, that is,
using the sequence
\[
\tau_0(k):=\sigma_k\wedge T,\quad \tau_{n+1}(k):=\inf\{t>\tau_{n}(k):\, |Y_t-Y_{\tau_n}|\geq \varepsilon \}\wedge T,
\] 
where $Y_t:=S_t+W_t$ and $W_t$ is as in the proof of Theorem \ref{main1}.

Apply the argument of Proposition \ref{dublin} starting from $\sigma_k\wedge T$ instead of $0$, 
using Lemma \ref{infamous} for $M_n:=Y_{\tau_n(k)}$. Choosing $\varepsilon$ small
enough, we get that there is $Q(k)\sim P$ such that 
\[
E_{Q(k)} g\left(\sup_{\sigma_k\leq t\leq T}|S_t-\tilde{S}_t|\right)<\chi \; \;
\mbox{and}\;\;  dQ(k)/dP=\prod_{j=1}^{\infty}Z^{(k)}_j,
\]
where we define $\tilde{S}_t=E_Q[S_T\vert\mathcal{H}_t]$, for \emph{all} $t\in [0,T]$.
Here $Z^{(k)}_j$ is $\mathcal{H}_{\tau_j(k)}$-measurable for each $j\in\mathbb{N}$, corresponding to the $Z_j$
appearing in the proof of Lemma \ref{infamous}. We now check that
$S_t=\tilde{S}_t$ a.s. on $\{t\leq\sigma_k\wedge T\}$. Indeed, on this set
\begin{eqnarray*}
S_t-\tilde{S}_t &=&E_{Q(k)}[S_{t\wedge \sigma_k}-\tilde{S}_{t\wedge\sigma_k}|\mathcal{H}_{t\wedge\sigma_k}]\\
&=&E_{Q(k)}[S_{t\wedge \sigma_k}-{S}_{T\wedge\sigma_k}|\mathcal{H}_{t\wedge\sigma_k}] \\
&=&\frac{E[(dQ(k)/dP)[S_{t\wedge\sigma_k}-S_{T\wedge\sigma_k}]|\mathcal{H}_{t\wedge\sigma_k}]}{
E[dQ(k)/dP\vert\mathcal{H}_{t\wedge\sigma_k}]}\\  
&=&\frac{E[E[(dQ(k)/dP)[S_{t\wedge\sigma_k}-S_{T\wedge\sigma_k}]|\mathcal{H}_{T\wedge\sigma_k}]|\mathcal{H}_{t\wedge\sigma_k}]}
{E[E[dQ(k)/dP\vert\mathcal{H}_{T\wedge\sigma_k}]\vert\mathcal{H}_{t\wedge\sigma_k}]}\\
&=&\frac{E[E[dQ(k)/dP|\mathcal{H}_{T\wedge\sigma_k}](S_{t\wedge\sigma_k}-S_{T\wedge\sigma_k})|\mathcal{H}_{t\wedge\sigma_k}]}{E[E[dQ(k)/dP\vert\mathcal{H}_{T\wedge\sigma_k}]\vert\mathcal{H}_{t\wedge\sigma_k}]}\\
 &=&0,\\
\end{eqnarray*}
which follows from $\tilde{S}_{T\wedge\sigma_k}=S_{T\wedge\sigma_k}$, $E[dQ(k)/dP\vert\mathcal{H}_{T\wedge\sigma_k}]=1$,  and 
the martingale property of $S$ up to $\sigma_k$ under $P$.  Now note that $P(dQ(k)/dP=1)\geq P(\sigma_k\wedge T=T)\to 1$, $k\to\infty$, which implies that $dQ(k)/dP$
tends to $1$ in probability, hence  almost surely along a subsequence. Using Scheff\'e's theorem,
we can find $k$ with $\Vert Q(k)-P\Vert_{tv}<\chi$. The last statement is clear
since $\mathcal{G}_T$ is independent of $\mathcal{F}_T$.
\end{proof}

\begin{corollary}\label{lp} Let Assumption \ref{large} be in force.
Let $p\geq 1$ be arbitrary and let $S$ be a sticky local martingale. Then for all $\chi>0$
there exists $Q\sim P$ with $||Q-P||_{tv}<\chi$ and a $Q$-martingale $\tilde{S}$
with respect to the enlarged filtration $\mathbb{H}$ such that 
$$
E_Q\sup_{t\in [0,T]}|S_t-\tilde{S}_t|^p<\chi
$$
is satisfied.\hfill $\Box$
\end{corollary}

\begin{remark} {\rm A \emph{strict local martingale} is a local martingale which is not a martingale.
It is not difficult to construct sticky strict local martingales by using Proposition 3.8 (also Corollary 3.9) of Elworthy et al.  \cite{yor}. We can choose any continuous and nonincreasing $m: \mathbb{R}^+\rightarrow (0, 1]$
with $m(0)=1$ and let $M_t=1/R_{r^{-1}(m(t))}$, where $R_t$ is a 3-dimensional Bessel process starting from $1$ and $r(t)=E(1/R_t)$. Then from Proposition 3.8 of Elworthy et al. \cite{yor}, $M_t$ is a strict local martingale with $m(t)=EM_t$.
($M_t$ is strict local martingale as long as 
$m(t)$ is not a constant).
$M_t$ is sticky as it is obtained from a sticky process $R$ by transformation under continuous function $\frac{1}{x}, x>0$  and by bounded time change.} 
\end{remark}

%\begin{remark}\label{surprise1} {\rm For strictly positive, sticky local martingales
%a modification of the construction in Theorem \ref{tetto} provides $\tilde{S}$ with
%$$
%E_Qg\left(\sup_{t\in [0,T]}\left|S_t/\tilde{S}_t-1\right|\right)<\chi.
%$$
%In the case where $S$ is continuous,
%even 
%$$
%1-\chi\leq \inf_{t\in [0,T]} S_t/\tilde{S}_t\leq \sup_{t\in [0,T]} S_t/\tilde{S}_t\leq 1+\chi
%$$ 
%can be guaranteed a.s. This latter conclusion has already been proved for a large class of diffusion processes in
%Guasoni and R\'asonyi \cite{fragile}. In that paper, however, no estimate for $||Q-P||_{tv}$ was 
%shown and discontinuous
%processes weren't treated either.}
%\end{remark}

\begin{remark} {\rm Strict local martingales (which are not martingales) 
have been suggested as models for financial bubbles, see Protter \cite{protter}. In this context, $P$
is the pricing measure, $S$ is the price process of risky assets. 
Theorem \ref{tetto} and Corollary \ref{lp} reiterate the word of caution already pronounced
in Guasoni and R\'asonyi \cite{fragile}: an arbitrarily small mis-specification of option and asset prices (that is,
mistaking $Q$ for $P$ and $\tilde{S}$ for $S$) may destroy the ``bubble phenomenon'' generated by
$S$ under $P$ since $\tilde{S}$ is a martingale under $Q$, admitting no bubbles.}
\end{remark}

\section{Application to mathematical finance}\label{appli}

A central concept of mathematical finance is arbitrage, i.e. riskless profit. 
Such opportunities should not exist in an efficient market. Arbitrage theory is
well-understood in idealized models of financial markets where the presence of frictions (transaction fees,
liquidity effects) is disregarded, see e.g. \cite{ds}. There is also a fairly clear
picture in the case of proportional transaction costs, where, roughly speaking, trading costs
are linear functions of the trading speed, see \cite{kabanov}. Illiquid markets,
however, show new phenomena due to a \emph{superlinear} dependence of trading costs on
the trading speed. In such market models a characterization for the absence of arbitrage 
in terms of dual variables
has been provided in the paper \cite{gr}, see Theorem \ref{char} below.

We will apply the results of the present paper to show that a very large class of
candidate price processes (namely the sticky ones) enjoy an absence of arbitrage property
in markets with superlinear liquidity effects. We now briefly sketch (a slightly simplified
version of) the model in Guasoni and R\'asonyi \cite{gr}, see \cite{gr} for further details.

Staying on the stochastic basis $(\Omega,\mathcal{F},P,\mathbb{F})$, let $S$ describe the price of $d$ risky assets in a financial market when trading is (infinitely) slow.
Liquidity effects will be described by a cost function 
$G:\Omega\times [0,T]\times \mathbb{R}^d\to\mathbb{R}_+$ 
which is assumed $\mathcal{O}\otimes\mathcal{B}(\mathbb{R}^d)$-measurable where
$\mathcal{O}$ is the optional sigma-field. We furthermore assume that
$G(\omega,t,\cdot)$ is convex with $G(\omega,t,x)\geq G(\omega,t,0):=0$ for all $\omega,t,x$. 
Henceforth, set $G_t(x):=G(\omega,t,x)$, i.e. the dependence on $\omega$ is omitted, and $t$ is used as a 
subscript. There is also a riskless asset $S^0_t$ of price constant $1$, $t\in [0,T]$.

A \emph{feasible strategy} is a process $\phi$ in the class
\begin{equation}
\mathcal{A}:=\left\{\phi:\phi\text{ is a $\mathbb{R}^d$-valued, optional process},
\int_0^T \vert\phi_u\vert du<\infty\text{ a.s.}\right\},
\end{equation}
i.e. the speed of trading $\phi_t$ at $t$ is assumed to be finite and the traded quantity of stocks over $[0,T]$
as well.

With this definition, for a given strategy $\phi\in\mathcal{A}$ and an initial asset position 
$z=(z^0,\ldots,z^d)\in\mathbb{R}^{d+1}$, the resulting positions at time $t\in [0,T]$ in the risky and safe assets are defined as:
\begin{eqnarray}\label{eq:w}
X^i_t(z,\phi) &:=& z^i+\int_0^t \phi_u^i du,\quad 1\le i\le d,\\
\label{eq:selffin}
X^0_t(z,\phi) &:=& z^0-\int_0^t \phi_u S_u du-\int_0^t G_u(\phi_u)du.
\end{eqnarray}

The main item in the following assumption is the superlinearity condition \eqref{eq:superlinear}: it 
expresses that fast trading has an effect which is stronger than linear as a function of 
the trading speed.
\begin{assumption}\label{below}
There is $\alpha>1$ and $H>0$ such that
\begin{eqnarray}
%\label{eq:hpositive}
%\inf_{t\in [0,T]} H_t &>&\ 0 \quad\text{a.s.},\\
\label{eq:superlinear}
G_t(x) &\geq&\ H\vert x\vert^{\alpha}, \quad\text{for all }\omega,t,x,\\
\label{eq:locint}
\int_0^T \left(\sup_{|x|\leq N}G_t(x) \right)dt &<&\ \infty \quad\text{a.s. for all }N>0.
\end{eqnarray}
\end{assumption}

Define also
\begin{equation*}
G_t^*(y) := \sup_{x\in\mathbb{R}^d} (xy-G_t(x))\geq 0,\ y\in\mathbb{R}^d,\ t\in [0,T].
\end{equation*}

We will apply results of Section \ref{mr} to investigate under which conditions such market models are
free of arbitrage. An \emph{arbitrage of the second kind} is a strategy $\phi\in\mathcal{A}$, 
such that $X_T^i(z,\phi)\geq 0$, $i=0,1,\ldots,d$ with $z=(c,0,\ldots,0)$  for some $c<0$. 
Absence of arbitrage of the second kind {(NA2)} holds if no such opportunity exists.

We reproduce Theorem 4.2 of Guasoni and R\'asonyi \cite{gr} below\footnote{Unfortunately,
the conditions ``$\{Z_t^i=0,\ i=1,\ldots,T\}\subset \{Z_t^0=0\}$
a.s. for all $t$, $Z_0^0=1$'' are missing from the statement of that theorem in Guasoni and R\'asonyi 
\cite{gr} (but they
are apparently needed in view of the preceding results there). 
Here in Theorem \ref{char} we state the corrected version.}
which characterizes (NA2). The notation $L^{p}(Q)$ for $p\geq 1$
refers to the usual Banach space of $d$-dimensional random variables with finite $p$th (absolute) moment
under the probability $Q\sim P$.
\begin{theorem}\label{char} 
Let $\mathcal{F}_0$ be trivial, let Assumption \ref{below} hold, fix $1<\beta<\alpha$ and let 
$1/\beta+1/\gamma=1$.
{\rm (NA2)} holds if and only if, for all $\chi>0$, 
there exists $Q\sim P$ with 
$$
E_Q \int_0^T (1+|S_t|)^{\beta\alpha/(\alpha-\beta)}dt <\infty 
$$
and an $\mathbb{R}^{d+1}_+$-valued 
$Q$-martingale $Z$ with $Z_T\in L^{\gamma}(Q)$ such that $\{Z_t^i=0,\ i=1,\ldots,T\}\subset \{Z_t^0=0\}$
a.s. for all $t$, $Z_0^0=1$ and
$E_Q\int_0^T Z^0_t G_t^*(\bar{Z}_t-S_t)dt<\chi$
where $\bar{Z}_t^i=(Z_t^i/Z^0_t) 1_{\{Z^0_t\neq 0\}}$, $i=1,\ldots,d$.
\hfill $\Box$
\end{theorem}

Theorem \ref{main1} ensures that a plethora of models satisfy (NA2).

\begin{proposition}\label{absence} Let $\mathcal{F}_0$ be trivial, let Assumption \ref{below} hold.
If $S$ is sticky then it satisfies (NA2).  
\end{proposition}
\begin{proof} We enlarge the probability space so that Assumption \ref{large} holds.
Lemma 3.2 of Guasoni and R\'asonyi \cite{gr} shows that there is a constant $C$ such that, for all $t$,
$G_t^*(y)\leq C|y|^{\alpha/(\alpha-1)}$. Set $\delta:=\max\{\gamma,\beta\alpha/(\alpha-\beta)\}$, 
$g(x):=x^{\delta}$, 
$x\geq 0$ and notice that $\alpha/(\alpha-1)\leq \gamma$ hence, for any random variable $X$,
\begin{equation}\label{ex}
E_QG_t^*(X)\leq CE_Q|X|^{\alpha/(\alpha-1)}\leq CE_Q^{\alpha/[(\alpha-1)\delta]}[|X|^{\delta}]. 
\end{equation}

Fix $\chi>0$. Theorem \ref{main1}  
provides $Q\sim P$ and
a $Q$-martingale $\tilde{S}$ with respect to $\mathbb{H}$ such that 
\begin{equation}\label{tak}
E_Qg(\sup_{t\in [0,T]}|S_t-\tilde{S}_t|)<\chi. 
\end{equation}
A closer look at the details of those arguments
shows that Lemma \ref{infamous} is used with the choice
$w(x)=4^{2\delta}|x|^{2\delta}+2|x|$. Remark \ref{neige} then
yields
\[
E_Q\sup_{k\in\mathbb{N}}|M_{k}|^{2\delta}<\infty 
\]
hence also
\[
E_Q g\left(\sup_{k\in\mathbb{N}}|M_{k}|\right)<\infty. 
\]
We thus get
\begin{equation}
E_Qg\left(\sup_{t\in [0,T]}|Y_t|\right)\leq E_Qg\left(\sup_{k\in\mathbb{N}}|Y_{\tau_k}|+\varepsilon\right)=
E_Qg\left(\sup_{k\in\mathbb{N}}|M_k|+\varepsilon\right)<\infty, 
\end{equation}
noting convexity of $g$. Since $|Y_t-S_t|<\varepsilon$, this implies
\begin{equation}\label{majdd}
E_Qg\left(\sup_{t\in [0,T]}|S_t|\right)<\infty,
\end{equation}
and hence also
\begin{equation}\label{androttir}
E_Qg\left(\sup_{t\in [0,T]}|\tilde{S}_t|\right)<\infty, 
\end{equation}
by \eqref{tak}.
Noting $\beta\alpha/(\alpha-\beta)\leq \delta$ and \eqref{majdd},
$$
E_Q \int_0^T (1+|S_t|)^{\beta\alpha/(\alpha-\beta)}dt <\infty.
$$

Define the $Q$-martingale $Z^0_t:=1$, $Z^i_t:=\tilde{S}_t^i$, $i=1,\ldots,d$.
Clearly, $Z_T\in L^{\gamma}(Q)$ by \eqref{androttir} and $\delta\geq\gamma$. We deduce from \eqref{ex} and \eqref{tak} that 
\begin{eqnarray*}
& & E_Q\int_0^T Z^0_t G_t^*(\bar{Z}_t-S_t)dt=\int_0^T E_Q G_t^*(\tilde{S}_t-S_t)dt\\
&\leq& TC E_Q^{\alpha/[(\alpha-1)\delta]}\left[\sup_{t\in [0,T]}g(|\tilde{S}_t-S_t|)\right]
\leq TC \chi^{\alpha/[(\alpha-1)\delta]},
\end{eqnarray*}
which goes to $0$ as $\chi\to 0$. This implies (NA2) for the class $\mathcal{A}$ defined with $\mathbb{H}$-optional
processes. As $\mathbb{F}$ is a subfiltration of $\mathbb{H}$, the result
follows for $\mathcal{A}$ defined with $\mathbb{F}$-optional processes. This finishes the proof.
\end{proof}

\begin{remark} {\rm Property (NA2) was established in Guasoni and R\'asonyi \cite{gr} for the class of continuous processes
$S$ satisfying the CFS-$O$ property, see Remark \ref{utal}. Using arguments of Bender et al. \cite{bps}, one could
establish (NA2) for continuous sticky processes $S$. The essential novelty of Proposition
\ref{absence} thus lies in allowing jumps for $S$.}
\end{remark}
 
\section{Auxiliary results}\label{appendix}

For the proof of Theorem \ref{main1}, we need the two Lemmas presented below. We 
fix some notations first. Scalar products in $\mathbb{R}^d$ will be denoted by $\langle \cdot,\cdot\rangle$. 
$\overline{\mathbb{R}^d}$ denotes the one-point compactification of $\mathbb{R}^d$ and $C(\overline{\mathbb{R}^d})$
denotes the set of $\mathbb{R}$-valued continuous functions on $\overline{\mathbb{R}^d}$. We let
$C_+(\overline{\mathbb{R}^d}):=\{g\in C(\overline{\mathbb{R}^d}):\ g(x)>0,\ x\in\mathbb{R}^d\}$. We denote by
$C_0(\mathbb{R}^d)$ the family of continuous functions with compact support on $\mathbb{R}^d$. 
As $\overline{\mathbb{R}^d}$ is compact, $C(\overline{\mathbb{R}^d})$ (equipped with
the supremum norm) is a separable Banach space, a fortiori a Polish space. 
As $C_+(\overline{\mathbb{R}^d})$ is clearly
a Borel subspace of $C(\overline{\mathbb{R}^d})$, the measurable selection theorem (see e.g. 
III. 44-45. in \cite{dm}) applies to multifunctions with values in $C_+(\overline{\mathbb{R}^d})$.
Fix a continuous function $w:\mathbb{R}^d\to\mathbb{R}_+$ with $w(0)=0$.  

The next Lemma will provide a positive function $f$ that is used for a
measure change with density $f(Y)$ in the arguments of Lemma \ref{infamous}. The idea here is
that, due to \eqref{sarah} below (which will be a consequence of stickiness in our applications
of Lemma \ref{baba}), one can guarantee
that the ``mean'' $EYf(Y)$ is $0$ while the ``norm'' $Ef(Y)w(Y)$ stays small, together
with the ``mass'' $Ef(Y)1_{\{|Y|\geq \eta\}}$ allocated outside a small ball.

\begin{lemma}\label{baba}  Let $Y$ be an $\mathbb{R}^d$-valued random variable with $0\in\mathcal{S}(Y)$.  
Assume that 
\begin{equation}\label{sarah}
P(Y\in B(0,\epsilon))>0             
\end{equation}
for all $\epsilon>0$. Then for each $\eta>0$ there exists 
$f\in C_+(\overline{\mathbb{R}^d})$ such that $Ef(Y)=1$, $Ef(Y)w(Y)<\eta$, 
$Ef(Y)1_{\{|Y|\geq \eta\}}<\eta$, and $Ef(Y)Y=0$. 
\end{lemma}
\begin{proof} Define $\tilde{w}(x):=w(x)+|x|$ and note that it suffices to show the result for $\tilde{w}$
instead of $w$.
Since $\mathcal{S}(Y)\subset\mathcal{D}(Y)$, $\mathcal{D}(Y)$ is a nonempty linear subspace of $\mathbb{R}^d$. If $\mathcal{D}(Y)=\{ 0\}$
then we  set $f(y):=1$  for all $y\in\mathbb{R}^d$. From now on we assume that $\mathcal{D}(Y)$ has
dimension at least $1$. Define 
\[
\mathcal{A}:=\{r\in C_+(\overline{\mathbb{R}^d}):\ Er(Y)\tilde{w}(Y)<\eta/2,\ Er(Y)<\eta\}.
\] 
Now set $A:=\{Er(Y)Y:\ r\in\mathcal{A}\}$. Clearly, $A\subset\mathcal{D}(Y)$   is a convex and nonempty set. To see this observe that  for any $h\in C_0(\mathbb{R}^d)$, $h\geq 0$,
\begin{equation}\label{forma}
r(y)=\kappa_1 h(y)+\kappa_2 e^{-\tilde{w}(y)}
\end{equation}
lies in $\mathcal{A}$ for $\kappa_1,\kappa_2>0$ small enough.

Denoting by $ri_D(A)$ the interior of $A$ in the relative topology of $\mathcal{D}(Y)$ we claim
that $0\in ri_D(A)$. If this were not true then there would exist a non-zero $l\in\mathcal{D}(Y)$ such that
$\langle l,a\rangle\geq 0$ for all $a\in A$. This implies
\begin{equation}\label{mark}
E\langle l,Y\rangle r(Y)\geq 0 
\end{equation}
for all $r$ of the form \eqref{forma} (with $\kappa_1,\kappa_2$ small enough).  
We can let $\kappa_2 \rightarrow 0$ and obtain that (\ref{mark}) also holds for all 
$r(y)=\kappa_1 h(y)$ with $h\in C_0(\mathbb{R}^d)$, $h\geq 0$. This clearly implies  
that $\langle l,Y\rangle\geq 0$  a.s.  Therefore from $0\in\mathcal{S}(Y)$
we obtain that  $\langle l,Y\rangle=0$ a.s., by Theorem 3 in \cite{jacod-shiryaev} so
$\langle l,z\rangle=0$ for all $z\in\mathcal{D}(Y)$. 
But $\langle l, l\rangle >0$ and we arrive at a contradiction.

It follows that $B(0,\delta)\cap \mathcal{D}(Y)\subset A$ for some $0<\delta<\eta$.
We choose $\delta>0$  small enough such that $\sup_{|y|\leq \delta} \tilde{w}(y)\leq\eta/2$.
Let us now take $m\in C_0(\mathbb{R}^d)$ that is positive in the interior
of $B(0,\delta)$, vanishes elsewhere,  and satisfies $Em(Y)=1$.  Such a function
exists since $P(Y\in B(0,\delta/2))>0$.

Set $c:=Em(Y)Y\in B(0,\delta)\cap\mathcal{D}(Y)$. There exists $r\in\mathcal{A}$ with $Er(Y)Y=-c$,
so setting $f(y):=(r(y)+m(y))/E[r(Y)+m(Y)]$ we have $Ef(Y)Y=0$.  Obviously,
$Em(Y)\tilde{w}(Y)<\eta/2$ and $E[r(Y)+m(Y)]>1$ hence $Ef(Y)\tilde{w}(Y)<\eta$, using the definition of $\mathcal{A}$.
It remains to check that $$
Ef(Y)1_{\{|Y|\geq \eta\}}=Er(Y)1_{\{|Y|\geq \eta\}}/E[r(Y)+m(Y)]\leq Er(Y)1_{\{|Y|\geq \eta\}}<\eta,
$$
which follows from  $\delta<\eta$ and the definition of $\mathcal{A}$.
\end{proof}

Now consider a sigma-algebra $\mathcal{K}\subset\mathcal{F}$ and an $\mathbb{R}^d$-valued
random variable $X$. Let $Q:\mathcal{B}(\mathbb{R}^d)\times
\Omega\to [0,1]$ be the conditional law of $X$ with respect to  $\mathcal{K}$. We denote by $\delta_0$ the
Dirac measure at the origin. The following Lemma  is a  ``kernel version'' of Lemma \ref{baba}
above.

\begin{lemma}\label{lema} Let $0\in\mathcal{S}(Q(\cdot,\omega))$ for a.s. $\omega$ and let $Q(B(0,\epsilon),\omega)>0$ hold
 a.s. for each $\epsilon>0$. Then for each $\eta>0$,  there is a $\mathcal{B}(\mathbb{R}^d)\otimes\mathcal{K}$-measurable
 $j:\mathbb{R}^d\times\Omega\to (0,\infty)$ such that for almost  all $\omega$ the following holds:
 \begin{eqnarray*}
&&\int_{\mathbb{R}^d} j(z,\omega)Q(dz,\omega) =1,\\ 
&&\int_{\mathbb{R}^d} j(z,\omega)z Q(dz,\omega)= 0,\\
&&\int_{\mathbb{R}^d} j(z,\omega)w(z) Q(dz,\omega) < \eta, \\
&&\int_{\mathbb{R}^d} j(z,\omega)1_{\{|z|\geq \eta\}}Q(dz,\omega) < \eta.\\
\end{eqnarray*}
Furthermore, we may choose $j(z,\omega):=1$, $z\in\mathbb{R}^d$ on $\{\omega:Q(\omega,\cdot)=\delta_0(\cdot)\}$. 
\end{lemma}
\begin{proof} For a.e. $\omega$, one can apply
Lemma \ref{baba} to a random variable $Y$ that has law $Q(\cdot,\omega)$ to get a function 
$f_{\omega}\in C_+(\overline{\mathbb{R}^d})$.
We may apply the measurable selection theorem on $(\Omega,\mathcal{K},P)$
to get a mapping $\omega\to f_{\omega}$ that is $\mathcal{K}/\mathcal{B}(C(\overline{\mathbb{R}^d}))$-measurable.
Let $j(z,\omega):=f_{\omega}(z)$ which is $\mathcal{B}(\mathbb{R}^d)\otimes\mathcal{K}$-measurable
(since each $f_{\omega}$ is continuous). This clearly satisfies the conclusions of the present lemma.
The last statement is clear from the proof of Lemma \ref{baba}.
\end{proof}

\begin{remark}\label{selection} {\rm The technology used in Lemmas
\ref{baba} and \ref{lema} was initiated in Dalang, et, al. \cite{dmw}. It has been further developed
by Y. Kabanov
in a continuous-time context and found several applications in mathematical finance. Here we only refer 
to Kabanov and Stricker \cite{ks_bounded} as a representative example.}
\end{remark}

Finally, we present a simple Lemma that was useful for the proof of Lemma \ref{sticky-lem} above.

\begin{lemma}\label{egyszeru}
Let $\mathcal{A},\mathcal{B}$ be sigma-fields and let $U,V$ be nonnegative random variables such that
$\mathcal{A}\vee\sigma(U)$ is independent of $\mathcal{B}\vee\sigma(V)$. Then
$$
E[UV\vert\mathcal{A}\vee\mathcal{B}]=E[U\vert\mathcal{A}]E[V\vert\mathcal{B}].
$$
\end{lemma}
\begin{proof}
Let $A$ (resp. $B,C,D$) be $\sigma(U)$ (resp. $\sigma(V),\mathcal{A},\mathcal{B}$) measurable sets.
By the monotone class theorem and by the definition of conditional expectations, it suffices to prove that
\begin{equation}\label{lapp}
E[1_A1_B 1_C1_D]=E[E[1_A\vert \mathcal{A}]E[1_B\vert\mathcal{B}]1_C1_D].
\end{equation}
By independence and by the $\mathcal{A}$ (resp. $\mathcal{B}$) measurability of $C$ (resp. $D$), 
the right-hand side of \eqref{lapp} is
$$
E[E[1_A1_C\vert\mathcal{A}]]E[E[1_B1_D\vert\mathcal{B}]]=E[1_A1_C]E[1_B1_D],
$$
which equals the left-hand side of \eqref{lapp} by independence of $1_A1_C$ from $1_B1_D$.
\end{proof}

\end{document}